\pdfoutput=1
\documentclass[9pt, format=sigplan, screen=true, review=false]{acmart}

\usepackage{macros}
\acmConference[LICS]{Logic in Computer Science}{July 9--12, 2018}{Oxford}

\title{Linear Equations with Ordered Data}

\author{Piotr Hofman}
\orcid{0000-0001-9866-3723}
\affiliation{University of Warsaw}
\email{piotrek.hofman@gmail.com}

\author{S\l awomir Lasota}
\orcid{}
\affiliation{University of Warsaw}
\email{sl@mimuw.edu.pl}


\begin{document}

\begin{abstract}
Following a recently considered generalization of linear equations to unordered data vectors, 
we perform a further generalization to \emph{ordered data} vectors.
These generalized equations naturally appear in the analysis of vector addition systems (or Petri nets) extended with ordered data.
We show that nonnegative-integer solvability of linear equations is computationally equivalent 
(up to an exponential blowup) with the
reachability problem for (plain) vector addition systems.
This high complexity is surprising, and contrasts with NP-completeness for unordered data vectors.
Also surprisingly, we achieve polynomial time complexity of the solvability problem when the
nonnegative-integer restriction on solutions is dropped.
\end{abstract}

\maketitle

\begin{acks}
This work is supported by the 
\grantsponsor{PL00000}{Polish National Science Centre}{https://www.ncn.gov.pl/?language=en} grants
\grantnum{PL00000}{2016/21/\linebreak[2]B/ST6/\linebreak[2]01505} (S.L.) and 
\grantnum[https://www.mimuw.edu.pl/~ph209519/LinAlgProj/]{PL00000}{
2016/21/\linebreak[2]D/ST6/01368} (P.H.).
\end{acks}


\section{Introduction}

Systems of linear equations are useful for approximate analysis of vector addition systems, or Petri nets.
For instance, the relaxation of semantics of Petri nets, where the configurations along a run are not required 
to be nonnegative, yields so called \emph{state equation} of a Petri net, which is a system of linear equations
with nonnegative-integer restriction on solutions. 
This is equivalent to \emph{integer linear programming}, a well-known 
\NP-complete problem~\cite{Karp21NP-complete-Problems}.
If the nonnegative-integer restriction if further relaxed to
nonnegative-rational one (or nonnegative-real one), we get a weaker but computationally more tractable approximation, equivalent to
\emph{linear programming} and solvable in polynomial time.
We refer to~\cite{SilvaTC96} for an exhaustive overview of linear-algebraic and integer-linear-programming techniques 
in analysis of Petri nets; usefulness of these techniques is confirmed by multiple applications including, for instance, recently proposed
efficient tools for the coverability problem of Petri 
nets~\cite{GeffroyLS16,BlondinFHH16}.

\para{Motivations}
A starting point for this paper is an extension of the model of Petri nets, or vector addition systems,
with data~\cite{LNORW08,HLLLST2016}.
This is a powerful extension of the model, which significantly enhances its expressibility but also increases the complexity of analysis.
In case of \emph{unordered} data (a countable set of data values that can be tested for equality only), 
the coverability problem is decidable (in non-elementary complexity)~\cite{LNORW08}
but the decidability status of the reachability problem remains still open.
In case of \emph{ordered} data, the coverability problem is still decidable while reachability is undecidable. 
(Petri nets with ordered data are equivalent to timed Petri nets, as shown in~\cite{rr-lsv-10-23}.)
One can also consider other data domains, and the coverability problem remains decidable as long as the 
data domain is homogeneous~\cite{Las16} (not to be confused with \emph{homogeneous} systems of linear equations), 
but always in non-elementary complexity.
In view of these high complexities, a natural need arises for efficient over-approximations.

A configuration of a Petri net with data domain $\D$ is a nonnegative integer
\emph{data vector}, i.e., a function $\D \to \N^d$ 
that maps only finitely many data values to a non-zero vector in $\N^d$. In a search for efficient over-approximations of
Petri nets with data, a natural question appears: Can linear algebra techniques 
be generalised so that the role of vectors is played by data vectors?
In case of unordered data, this question was addressed 
in~\cite{HLT2017LICS}, where first promising results has been shown, 
namely the 
nonnegative-integer solvability of linear equations over unordered data domain is \NP-complete.
Thus, for unordered data, the problem remains within the same complexity class as its plain (data-less) counterpart.
The same question for the second most natural data domain, i.e.~ordered data, 
seems to be even more important; ordered data enables modeling features 
like fresh names creation or time dependencies.

\para{Contributions}
In this paper we do a further step and investigate linear equations with ordered data, for which
we fully characterise the complexity of the solvability problem.
Firstly, we show that nonnegative-integer solvability of linear equations is computationally equivalent 
(up to an exponential blowup) with the reachability problem for plain Petri nets (or vector addition systems).
This high complexity is surprising, and contrasts with NP-completeness for unordered data vectors.
Secondly, we prove, also surprisingly, that the complexity of the solvability problem drops back to polynomial time, when
the nonnegative-integer restriction on solutions is relaxed to 
nonnegative-rational, integer, or rational.
Thirdly, we offer a conceptual contribution and notice that
systems of linear equations with (unordered or) ordered data are a special case of 
systems of linear equations which are infinite but finite up to an automorphism of data domain.
This recalls the setting of \emph{sets with atoms}~\cite{atombook,TMatoms,locfin}, with a data domain being a parameter, and 
the notion of \emph{orbit-finiteness} relaxing the classical notion of finiteness.

\para{Outline}
In Section~\ref{sec:vas} we introduce the setting we work in, and formulate our results.
Then the rest of the paper is devoted to proofs.
First, in Section~\ref{sec:lowerbound} we provide a lower bound for the nonnegative-integer solvability problem,
by a reduction from the VAS reachability problem.
Then, in Section~\ref{sec:hist} we suitably reformulate our problem 
in terms \emph{multihistograms}, which are matrices satisfying certain combinatorial property.
This reformulation is used In the next Section~\ref{sec:upperbound} to provide 
a reduction from the nonnegative-integer solvability problem to the reachability problem of vector addition systems, 
thus proving decidability of our problem.
Finally, in Section~\ref{sec:p} we investigate various relaxations of the 
nonnegative-integer restriction on solutions and work out 
a polynomial-time decision procedure in each case.
In the concluding Section~\ref{sec:orbitfinite} we sketch upon a generalised setting of orbit-finite systems of linear equations.

%
%


\section{Vector addition systems and linear equations} \label{sec:vas}

In this section we introduce the setting of linear equations with data, and formulate our results. 
For a gentle introduction of the setting, we start by recalling classical linear equations.

Let $\Q$ denote the set of rationals, and $\Qplus, \Z$, and $\N$ denote the subsets
of nonnegative rationals, integers, and nonnegative integers.
%
Classical linear equations are of the form
\[
a_1 x_1 + \ldots a_m x_m = a,
\] 
where $x_1 \ldots x_m$ are variables (unknowns), and $a_1 \ldots a_m\in \Q$ are rational coefficients.
For a system $\calU$ of such equations over the same variables $x_1, \ldots, x_m$, a solution of ${\calU}$
is a vector $(n_1, \ldots, n_m) \in \Q^m$ such that the valuation 
$x_1 \mapsto n_1, \ldots$, $x_m \mapsto n_m$ satisfies all equations in $\calU$.
In the sequel we are most often interested in nonnegative integer solutions $(n_1, \ldots, n_m) \in \N^m$,
but one may consider also other solution domains than $\N$.
It is well known that the \emph{nonnegative-integer solvability problem} ($\N$-solvability problem) 
of linear equations, i.e.~the question whether 
${\calU}$ has a nonnegative-integer solution, is NP-complete~\cite{taming}.
The complexity remains the same for other natural variants of this problem, for instance for inequalities instead of equations
(a.k.a.~integer linear programming).
On the other hand, for any $\X \in \{\Z, \Q, \Qplus\}$, the 
\emph{$\X$-solvability problem}, i.e., the
question whether $\calU$ has a solution $(n_1, \ldots, n_m) \in \X^m$,
is decidable in polynomial time.
%
%

\begin{remark}[integer coefficients] \label{rem:int}
A system of linear equations with rational coefficient can be transformed in polynomial time 
to a system of linear equations with integer coefficients, while preserving the set of solutions.
Thus from now on we allow only  for \emph{integer} coefficients $a_1 \ldots a_m$ in linear equations.
\end{remark}

The $\X$-solvability problem is equivalently formulated as follows: for a given finite set of coefficient vectors 
$A = \{\vec a_1, \ldots, \vec a_m\} \finsubseteq \Z^\dimension$ and a target vector $\vec a \in \Z^\dimension$ 
 (we use bold fonts to distinguish vectors from other elements), check whether 
$\vec a$ is an \emph{$\X$-sum} of $A$, 
i.e., a sum of the following form 
\begin{align} \label{eq:sum}
\vec{a}=\sum_{i=1}^{m} c_i \cdot \vec a_i, 
\end{align}
for some $c_1, \ldots, c_m \in \X$.
The number $d$ corresponds to the number of equations in $\calU$ and is called the \emph{dimension} of $\calU$.

Linear equations may serve as an over-approximation of 
the reachability set of a Petri net, or equivalently, of a \emph{vector addition system} -- we prefer to work with the latter model. 
A vector addition system (VAS) 
${\cal V} = (A, \vec{i}, \vec{f})$ is defined, similarly as above, by a finite set of vectors
$A \finsubseteq \Z^\dimension$ together with two nonnegative vectors $\vec{i}, \vec{f} \in \N^\dimension$, the initial one and the final one.
The set $A$ determines a transition relation $\longrightarrow$ between configurations, which are nonnegative integer vectors $\vec{c}\in \N^\dimension$:
there is a transition
$\vec{c} \longrightarrow \vec{c}'$ if $\vec{c}' = \vec{c} + \vec{a}$ for some $\vec{a}\in A$.
The VAS reachability problems asks, whether the final configuration is reachable from the initial one by a sequence of transitions,
i.e. $\vec{i} \longrightarrow^* \vec{f}$. 
It is important to stress that intermediate configurations are required to be nonnegative.
In other words, the reachability problem asks whether there is a sequence $\vec{a}_1, \vec{a}_2, \ldots, \vec{a}_m \in A$ 
(called a \emph{run}) such that 
%
\begin{align*}
\vec{i} + \sum_{i = 1}^m \vec{a}_i = \vec{f} && \qquad
\vec{i} + \sum_{i = 1}^j \vec{a}_i \geq \vec{0}, \text{ for every } j \in \{1\dots m\}
\end{align*}
\noindent
where $\zerovec$ denotes a zero vector (its length will be always clear from the context).
The problem is decidable~\cite{mayr81,kosaraju82} and \expspace-hard~\cite{Lipton76},
and nothing is known about complexity except for the cubic Ackermann upper bound of~\cite{demystifying}.
For a given VAS, a necessary condition for reachability is that
$
\vec{f} - \vec{i} =\sum_{i=1}^{m} \vec{a_i},
$
which is equivalent to $\N$-solvability of a system of linear equations, called (in case of Petri nets) 
the \emph{state equation}.
For further details we refer the reader to an exhaustive overview of linear-algebraic approximations for
Petri nets~\cite{SilvaTC96}, where both $\N$- and $\Qplus$-solvability problems are considered.


\subsection{Vector addition systems and linear equations, with ordered data}
\label[section]{sec:datavectors}

The model of VAS, and linear equations, can be naturally extended with data.
In this paper we assume that the data domain $\setD$ is a countable set,
ordered by a dense total order $\leq$ with no minimal nor maximal element.
Thus, up to isomorphism, $(\setD, \leq)$ is rational numbers with the natural ordering.
Elements of $\setD$ we call \emph{data values}.
In the sequel we use order preserving permutations (called \emph{data permutations} in short) of $\setD$, 
i.e.~bijections $\rho : \setD\to\setD$ such that $x\leq y$ implies $\rho(x)\leq \rho(y)$.

A \emph{data vector} is a function $\vec{v}:\setD\rightarrow \Q^\dimension$
such that the \emph{support}, i.e.~the set
$\support{\vec{v}}\eqdef \{\alpha\in \setD \mid \vec{v}(\alpha)\not=\zerovec\}$,
is finite (similarly as for vectors, we use bold fonts to distinguish data vectors from other elements).
The vector addition $+$ is lifted to data vectors pointwise, so that
$(\vec{v}+\vec{w})(\alpha) \eqdef \vec{v}(\alpha)+\vec{w}(\alpha)$.
A data vector $\vec{v}$ is \emph{nonnegative} if $\vec{v} : \setD\to(\Qplus)^\dimension$, 
and $\vec v$ is \emph{integer} if $\vec{v} : \setD\to\Z^\dimension$.

Writing $\circ$ for function composition, we see that $\vec{v}\circ\rho$ is a data vector for any data vector 
$\vec{v}$ and any order preserving data permutation $\rho:\setD\to\setD$.
For a set $\setV$ of data vectors we define
\[
\perm{\setV} = \setof{\vec v \circ \rho}{\vec v \in \setV, \rho \text{ a data permutation}}.
\]
A data vector $\target$ is said to be a \emph{\pproduct} of a finite set of data vectors $\setV$
if there are $\vec{v_1},\ldots,\vec{v_m} \in \perm{\setV} $, not necessarily pairwise different, such that
$
\target=\sum_{i=1}^{m} \vec{v}_i. 
$
%
In the generalisations of the classical solvability problem, to be defined now, 
we allow as input only integer data vectors (cf.~Remark~\ref{rem:int}):

\decproblem{\pproductprob}
{a finite set $\setV$ of integer data vectors and an integer data vector $\target$}
{is $\target$ a \pproduct\ of $\setV$?}
      
\noindent
In the special case when the supports of $\target$ and all vectors in $\setV$ are all singletons, the \pproductprob is just
$\N$-solvability of linear equations and thus the \pproductprob is trivially \NP-hard.
As the first main result, we prove the following inter-re\-du\-ci\-bi\-li\-ty:
\begin{theorem}\label[theorem] {thm:eqVASreach}
The \pproductprob and the VAS reachability problem are inter-reducible, with 
an exponential blowup. 
\end{theorem}
%
    
Our setting generalises the setting of \emph{unordered} data, where the data domain $\setD$ is \emph{not} ordered, and
hence data permutations are all bijections $\setD\to\setD$. In the case of unordered data the \pproductprob 
is \NP-complete, as shown in~\cite{HLT2017LICS}.
The increase of complexity caused by the order in data is thus remarkable.

Similarly as linear equations in the data-less setting, 
\pproductprob may be used as an overapproximation of the reachability in
vector addition systems with ordered data, which are defined 
exactly as ordinary VAS but in terms of data vectors instead of ordinary vectors. 
A VAS with ordered data ${\cal V} = (\setV, \vec{i}, \vec{f})$ consists of $\setV\finsubseteq 
\setD\to \Z^\dimension$ a finite set  of integer data vectors, and the initial and final
nonnegative integer data vectors $\vec{i}, \vec{f} \in \setD\to \N^d$. The configurations 
are nonnegative integer data vectors, and the set $\setV$ induces
a transition relation between configurations as follows: $\vec{c} \longrightarrow \vec{c'}$ if $\vec{c'} = \vec{c} + \vec{v}$ 
for some $\vec{v} \in \perm{\setV}$.
The reachability problem asks whether the final configuration is reachable from the initial one by a sequence
of transitions, $\vec{i} \longrightarrow^* \vec{f}$; it is undecidable~\cite{LNORW08}.
(The decidability status of the reachability problem for VAS with \emph{unordered} data is unknown.)
As long as reachability is concerned, 
VAS with (un)ordered data are equivalent to Petri nets with (un)ordered data~\cite{HLLLST2016}.

The \pproductprob is easily generalised to other domains $\X\subseteq\Q$ of 
solutions.
To this end we introduce scalar multiplication: 
for $c\in \Q$ and a data vector $\vec{v}$ we put 
$(c\cdot \vec v)(\da)\eqdef c\vec{v}(\da)$.
A data vector $\target$ 
is said to be a \emph{\xpproduct{$\X$}} of a finite set of data vectors $\setV$ if there are
$\vec{v_1},\ldots,\vec{v_m} \in \perm{\setV}$, not necessarily pairwise different, 
and coefficients $c_1,c_2\ldots c_m \in \X$
such that  (cf.~\eqref{eq:sum}) 
\[
\target=\sum_{i=1}^{m} c_i \cdot \vec{v_i}. 
\]
%
This leads to the following version of {\sc \pproductprob} parametrised by the choice of solution domain $\X$:

\decproblem{\xpproductprob{$\X$}}
{a finite set $\setV$ of integer data vectors and an integer data vector $\target$}
{is $\target$ an \xpproduct{$\X$} of $\setV$?}
      
\noindent
The \pproductprob is a particular case, for $\X = \N$.
Our second main result is the following:
\begin{theorem}\label[theorem] {thm:p}
For any $\X \in \{\Z, \Q, \Qplus\}$, the \xpproductprob{$\X$} is in \ptime.
\end{theorem}

\noindent
For $\X\in \{\Z, \Q\}$, the above theorem is a direct consequence of a more general fact, where
$\Q$ or $\Z$ is replaced by any commutative ring $\R$, under a proviso that
data vectors are defined in a more general way, as finitely supported functions $\D \to \R^d$.
With this more general notion, we prove that 
the \xpproductprob{$\R$} reduces in polynomial time to the $\R$-solvability of linear equations with coefficients from $\R$
(cf.~Theorem~\ref{thm:pp} in Section~\ref{sec:pp}).

The case $\X=\Qplus$ in \Cref{thm:p} is more involved but of particular interest, as it recalls
continuous Petri nets \cite{serge, serge-compl} where fractional firing of a transition is allowed,
and leads to a similar elegant theory and efficient algorithms based on $\Qplus$-solvability of linear equations.
Moreover, faced with the high complexity of Theorem~\ref{thm:eqVASreach}, 
it is expected that Theorem~\ref{thm:p} may become a cornerstone of linear-algebraic techniques for VAS with ordered data.


\section{Lower bound for the \pproductprob} \label[section]{sec:lowerbound}

In this section, all data vectors are silently assumed to be integer data vectors.
We are going to show a reduction from the VAS reachability problem to the \pproductprob.
Fix a VAS ${\cal A} = (A, {\vec i}, {\vec f})$. We are going to define a set of data vectors $\setV$ and a target data vector $\target$
such that the following conditions are equivalent:

\begin{enumerate}
\item $\vec f$ is reachable from $\vec i$ in $\cal A$;
\item $\target$ is a \pproduct\ of $\setV$.
\end{enumerate}

\noindent
W.l.o.g.~assume $\vec f = \zerovec$.

We need some auxiliary notation.
First, note that every integer vector ${\vec a} \in \Z^d$ is uniquely presented as a difference $\vec{a} = \vec{a}^+ - \vec{a}^-$ 
of two nonnegative vectors $a^+ \in \N^d$ and
$a^- \in N^d$ defined as follows:
\begin{align*}
a^+(i) = \begin{cases}
a(i), & \text{ if } a(i) \geq 0 \\
0, & \text{ if } a(i) < 0 
\end{cases}
&& \ 
a^-(i) = \begin{cases}
- a(i), & \text{ if } a(i) \leq 0 \\
0, & \text{ if } a(i) > 0 
\end{cases}
\end{align*}
For a nonnegative vector ${\vec a} \in\N^d$, by a \emph{data spread} of $\vec a$ we mean any nonnegative integer data vector 
$\vec v : \setD \to \N^d$ such that 
\[
\sum_{\da\in\support{\vec v}} {\vec v}(\da) = {\vec a}.
\]
In words, for every coordinate $i$, the value ${\vec a}(i)\geq 0$ is spread among all values ${\vec v}(\da)(i)$, for all data values $\da \in \setD$;
clearly, ${\vec v}$ is finitely supported.


The rough idea of the reduction is to simulate every transition ${\vec a} \in A$ by a data spread of $\vec a$
such that, intuitively, all positive numbers in $\vec a$ use larger data values than all negative values.
By a \emph{data realization} of a vector $\vec{a}\in A$ we mean any data vector of the form
$\vec{v} = {\vec s}^+ - {\vec s}^-$, where data vector ${\vec s}^-$ is a data spread of ${\vec a}^-$, 
data vector ${\vec s}^+$ is a data spread
of ${\vec a}^+$, and $\support{{\vec s}^-} < \support{{\vec s}^+}$
(with the meaning that every element of $\support{{\vec s}^-}$ is smaller than every element of $\support{{\vec s}^+}$).
Intuitively, the effect of ${\vec v}$ is like the effect of $\vec a$ but additionally data values involved are \emph{increased}.
We will shortly write $\supportplus{\vec v}$ for $\support{\vec{s}^+}$ and $\supportminus{\vec v}$ for $\support{\vec{s}^-}$.
Clearly, a non-zero vector $\vec{a}$ has infinitely many different data realizations; on the other hand, 
there are only finitely many of them \emph{up to data permutation}. 
Let $V_\vec{a}$ be a set of data realizations of $\vec a$ containing representatives 
up to data permutation. 
The cardinality of $V_\vec{a}$ is exponential with respect to the size of $\vec{a}$.

Now we are ready to define $\setV$ and $\target$: we put $\setV = \bigcup_{{\vec a}\in A} V_{\vec a}$, and
as the target vector $\target$ we take $\target = -\vec{\bar i}$, for some arbitrary data spread $\vec{\bar i}$ of
$\vec i$ (recall that $\vec{\bar f} = \zerovec$).
 
It remains to prove the equivalence of conditions 1.~and 2. First, 1.~easily implies 2.~as every run of $\cal A$ can be transformed into 
a \pproduct\ of $\setV$ that sums up to $\target$, using suitable data realisations of the vectors used in the run. 

For the converse implication, suppose that
$\target=\sum_{i=1}^{n} {\vec w}_i$, where $\vec{w}_i 
  = \vec{v}_i \circ\theta_{i}$ and $\vec{v}_i \in \setV$.
By construction of $\setV$, for every $i \leq n$ the data vector $\vec{v}_i$ belongs to $V_{\vec{a}_i}$ for some $\vec{a}_i \in A$.
We claim that the multiset of vectors $\{\vec{a}_i\}_{i = 1}^n$ can be arranged into a sequence being a correct run of the VAS $\cal A$
from $\vec{i}$ to $\vec{f}$.
For this purpose we define a binary relation of \emph{immediate consequence} on data vectors $\vec{w}_i$: we say that
$\vec{w}_j$ is an immediate consequence of $\vec{w}_i$ if the intersection of $\supportplus{\vec{w}_i}$ and 
$\supportminus{\vec{w}_j}$ is non-empty. We observe that the reflexive-transitive closure of the immediate consequence
is a partial order. Indeed, antisymmetry follows due to the fact that all data vectors $\vec{w}_i$ satisfy
$\supportminus{\vec{w}_i} < \supportplus{\vec{w}_i}$.
Let $\preceq$ denote an arbitrary extension of the partial order to a total order, and suppose w.l.o.g.~that
$$\vec{w}_1 \prec \vec{w}_2 \prec \ldots \prec \vec{w}_n.$$
We should prove that the corresponding sequence $\vec{a}_1 \vec{a}_2 \ldots \vec{a}_n$ of vectors from $A$ is a correct run of
the VAS $\cal A$ from $\vec{i}$ to $\vec{f}$. 
This will follow, once we demonstrate that the sequence $\vec{w}_1 \vec{w}_2 \ldots \vec{w}_n$ is a correct
run in the VAS with ordered data with transitions $\setV$ and the initial configuration $\vec{\bar i}$.
We need to prove the data vector
$
\vec{u}_i = \vec{\bar i} + \sum_{i = 1}^{j} \vec{w}_i
$
is nonnegative for every $j \in \{0,\ldots, n\}$.
To this aim fix $\da \in \setD$ and $l \in \{1, \ldots, d\}$, and consider the sequence of numbers
\begin{align} \label{eq:seq}
\vec{u}_0(\da, l), \quad \vec{u}_1(\da, l), \quad \ldots \quad \vec{u}_n(\da, l)
\end{align}
appearing as the value of the consecutive data vectors $\vec{u}_0$, $\vec{u}_1$, $\ldots$, $\vec{u}_n$ 
at data value $\da$ and coordinate $l$.
We know that the first element of the sequence $\vec{u}_0(\da, l) = \vec{i}(\da, l) \geq 0$ and 
the last element of the sequence $\vec{u}_n(\da, l) = \vec{f}(\da, l) \geq 0$.
Furthermore, by the definition of the ordering $\preceq$ we know that the sequence~\eqref{eq:seq} is first non-decreasing, and 
then non-increasing.
These conditions imply nonnegativeness of all numbers in the sequence.

\begin{remark}
The exponential blowup in the reduction is caused only by binary encoding of numbers in vector addition systems;
it can be avoided if numbers are assumed to be encoded in unary 
or, equivalently, if instead of vector addition systems one uses counter machines without zero tests.
\end{remark}


\section{Histograms}  \label[section]{sec:hist}

The purpose of this section is to transform the \pproductprob to a more manageable form.
As the first step, we eliminate data by rephrasing the problem in terms of matrices.
Then, we distinguish matrices with certain combinatorial property, 
called \emph{histograms}, and use them to further simplify the problem.
In Lemma~\ref{lem:first} at the end of this section
we provide a final characterisation of the problem, using \emph{multihistograms}.
The characterisation will be crucial for effectively solving the \pproductprob in the following Section~\ref{sec:toVAS}.

In this section, all matrices are integer matrices, and all data vectors are integer data vectors.

\para{Eliminating data}
Rational matrices with $r$ rows and $c$ columns we call 
$\matr{r}{c}$-matrices, and $r$ (resp.~$c$) we call 
row (resp.~column) dimension of an $\matr{r}{c}$-matrix. 
We are going to represent any data vector $\vec v$ as 
a $\matr{d}{\card{\support{\vec v}}}$-matrix $M_{\vec{v}}$ as follows:
if $\support{\vec v} = \set{\alpha_1 < \alpha_2 < \ldots < \alpha_n}$, we put
\[
M_{\vec v}(i, j) \ \eqdef \ \vec v(i)(\alpha_j).
\]
A \emph{0-extension} of an $\matr{r}{c}$-matrix $M$ is any $\matr{r}{c'}$-matrix $M'$, $c' \geq c$, 
obtained from $M$ by inserting arbitrarily $c' - c$ additional
zero columns $\zerovec \in \Z^r$. 
Thus row dimension is preserved by 0-extension, and column dimension may grow arbitrarily.
We denote by $\zeroext{M}$ the (infinite) set of all 0-extensions of a matrix $M$.
In particular, $M \in \zeroext{M}$.
For a set $\setM$ of matrices we denote by $\zeroext{\setM}$ the set of all 0-extensions of all matrices in $\setM$. 
%

\begin{example}  \label[example]{ex:matr}
For a data vector $\vec v$ with support $\support{\vec v} = \set{\alpha_1 < \alpha_2}$, defined by
$\vec v(\alpha_1) = (1, 3, 0) \in \Z^3$ and $\vec v(\alpha_2) = (2, 0, 2) \in \Z^3$, here is the corresponding
matrix and two its exemplary 0-extension:
\[
M_{\vec v} = 
\begin{bmatrix}
1 & 2 \\
3 & 0 \\
0 & 2
\end{bmatrix}
\qquad
\begin{bmatrix}
0 & 1 & 2  \\
0 & 3 & 0  \\
0 & 0 & 2 
\end{bmatrix}
,
\begin{bmatrix}
1 & 0 & 0 & 2 \\
3 & 0 & 0 & 0 \\
0 & 0 & 0 & 2
\end{bmatrix}
\in \zeroext{M_{\vec v}}
\]
\end{example}
%
%
Below, whenever we add matrices we silently assume that they have the same row and column dimensions.
For a finite set $\setM$ of matrices, we say that a matrix $N$ is a sum of 0-extensions of $\setM$ if
\begin{align} \label{eq:zeroextsum}
N  =  \sum_{i=1}^{m} M_i
\end{align}
for some matrices $M_1, \ldots, M_m \in \zeroext{\setM}$, 
necessarily all of the same row and column dimension.
We claim that the \pproductprob is equivalent to the question 
whether some 0-extension of a given matrix $\targetMatr$ is a sum of 0-extensions of $\setM$.

\decproblem{\zeroextprob}
{a finite set $\setM$ of matrices, and a target matrix $\targetMatr$, all of the same row dimension $d$}
{is some 0-extension of $\targetMatr$ a sum of 0-extensions of $\setM$?}

\begin{lemma} \label{lem:equiv}
The \pproductprob 
is polynomially equivalent to the \zeroextprob.
\end{lemma}
\begin{proof}
We describe the reduction of \pproductprob to the \zeroextprob.
(The opposite reduction
is shown similarly and is omitted here.)

Given an instance $\target, \setV$ of the former problem, we define the instance
\[
\targetMatr=M_\target, \quad \setM=\setof{M_\vec v}{\vec v \in \setV}
\]
of the latter one. We need to show that $\target$ is \pproduct of $\setV$ if, and only if 
some 0-extension $N$ of $\targetMatr$ is a sum of 0-extensions of $\setM$.
In one direction, suppose $\target$ is a \pproduct of $\setV$, i.e.,
\begin{align} \label{eq:pprodeq}
\target=\sum_{i=1}^{m} \vec{v}_i \circ\rho_{i}
\end{align}
and let $\set{\alpha_1 < \ldots < \alpha_c}$ be the union of all supports of data vectors $\vec{v}_i \circ \rho_i$
(thus also necessarily including the support of $\target$).
We will define a matrix $N$ and matrices $M_1, \ldots, M_m$, 
as required in~\eqref{eq:zeroextsum}, all of the same column dimension $c$.
Thus their columns will correspond to data values $\alpha_1, \ldots, \alpha_c$.
Let $N$ be the unique 0-extension of $M_\target$ of column dimension $c$ so that the nonempty columns
are exactly those corresponding to element of $\support{\target}$.
Similarly, let $M_i$ be the unique $0$-extension of $\vec{v}_i \circ \rho_i$ of column dimension $c$,
whose nonempty columns correspond to elements of $\support{\vec{v}_i \circ \rho_i}$.
The so defined matrices satisfy the equality~\eqref{eq:zeroextsum}.

In the other direction, suppose the equality~\eqref{eq:zeroextsum} holds for some matrices $N \in\zeroext{M_\target}$ and 
$M_1 \in \zeroext{M_{\vec v_1}}$ $\ldots$ $M_m\in\zeroext{M_{\vec v_m}}$,
and let $c$ be their common column dimension.
Choose arbitrary $c$ data values $\alpha_1 < \alpha_2 < \ldots < \alpha_c$ so that
$\support{\target} \subseteq \set{\alpha_1, \ldots, \alpha_c}$ corresponds to nonempty columns of $N$, and define data permutations
$\rho_1 \ldots \rho_m$ so that $\rho_i$ maps the support of $\vec v_i$ to data values corresponding to nonempty columns
in $M_i$. One easily verifies that~\eqref{eq:pprodeq} holds, 
as required.
\end{proof} 

From now on we concentrate on solving the 
\zeroextprob.

\para{Histograms}
We write briefly $\sum H(i, 1 \ldots j)$ as a shorthand for 
$\sum_{1\leq l\leq j}  H(i,l)$.
In particular, $\sum H(i, 1\ldots 0) = 0$ by convention.
An integer matrix we call nonnegative if it only contains nonnegative integers.
Histograms, to be defined now, are an extension of histograms of~\cite{HLT2017LICS} to ordered data.
\begin{definition} \label{def:hist}
A nonnegative integer $\matr{r}{c}$-matrix $H$ we call a \emph{histogram} if 
the following conditions are satisfied:
\begin{itemize}
\item there is $s>0$ such that $\sum H(i,1\ldots c)=s$ for every $1\leq i \leq  r$;
$s$ is called the \emph{degree} of $H$;
\item for every $1\leq i<r$ and $0\leq j < c$, the inequality holds:
$$\sum H(i, 1 \ldots j) \geq \sum H(i+1, 1\ldots j+1).$$
\end{itemize}
\end{definition}

\noindent
Note that the definition enforces $r \leq c$,~i.e., the column dimension $c$ of a histogram 
is at least as large as its row dimension $r$.
Indeed, forcedly
\begin{align*}
& H(2, 1) \ = \ 0 \\
& H(3, 1)  \ = \ H(3, 2) \ = \ 0 \\
& \ldots \\
& H(r, 1)  \ = \ \ldots  \qquad = \ H(r, r-1) \ = \ 0.
\end{align*}
Histograms of degree $1$ we call \emph{simple} histograms.
\begin{example}
A histogram of degree $2$ decomposed as a sum of two simple histograms:
\begin{align*}
&\begin{bmatrix}
1 & 1 & 0 & 0 & 0 \\
0 & 0 & 2 & 0 & 0 \\
0 & 0 & 0 & 1 & 1
\end{bmatrix}
\ = \  
\begin{bmatrix}
1 & 0 & 0 & 0 & 0 \\
0 & 0 & 1 & 0 & 0 \\
0 & 0 & 0 & 1 & 0
\end{bmatrix}
\ + \ 
\begin{bmatrix}
0 & 1 & 0 & 0 & 0 \\
0 & 0 & 1 & 0 & 0 \\
0 & 0 & 0 & 0 & 1
\end{bmatrix}
\end{align*}
\end{example}
The following combinatorial property of histograms will be crucial in the sequel:
\begin{lemma}\label[lemma]{lem:histogram}
$H$ is a histogram of degree $s$ if, and only if $H$ is a sum of $s$ simple histograms.
\end{lemma}

\noindent
Below, whenever we multiply matrices we silently assume that the column dimension of the first one is the same as 
the row dimension of the second one.
Simple histograms are useful for characterising 0-extensions: 
\begin{lemma} \label[lemma]{lem:simplehistsum}
For matrices $N$ and $M$, $N \in \zeroext{M}$ if, and only if
$N  =  M \cdot S$, for a simple histogram $S$.
\end{lemma}
\begin{example}
Recall the matrix $M = M_\vec v$ from~\cref{ex:matr}.
One of the matrices from $\zeroext{M}$ is presented as multiplication of $M$ and a simple histogram as follows:
\begin{align*}
\begin{bmatrix}
1 & 0 & 2 & 0 \\
3 & 0 & 0 & 0 \\
0 & 0 & 2 & 0
\end{bmatrix}
\quad & = &&
\begin{bmatrix}
1 & 2 \\
3 & 0 \\
0 & 2
\end{bmatrix}
\cdot
\begin{bmatrix}
1 & 0 & 0 & 0 \\
0 & 0 & 1 & 0 \\
\end{bmatrix}&
\end{align*}
\end{example}
We use Lemmas~\ref{lem:histogram}  and~\ref{lem:simplehistsum} to characterise the \zeroextprob:
\begin{lemma} \label[lemma]{lem:histsum}
For a matrix $N$ and a finite set of matrices $\setM$, the following conditions are equivalent:

\begin{enumerate}
\item \label{eq:1} $N$ is a sum of 0-extensions of $\setM$;
\item \label{eq:2} $N \ = \ \sum_{M \in \setM} M \cdot H_M$, for some histograms $\setof{H_M}{M\in \setM}$.
\end{enumerate}
\end{lemma}
\begin{proof}
In one direction, assume condition~\ref{eq:1}. holds, i.e.,
\begin{align} \label{eq:assumed}
N \ = \ \sum_\cokol N_\cokol  
\end{align}
for $N_\cokol \in \zeroext{M_\cokol}$, $M_\cokol \in \setM$, and then apply~\cref{lem:simplehistsum} to get
(simple) histograms $H_\cokol$ with $N_\cokol = M_\cokol \cdot H_\cokol$. 
Thus
%
$
N  =  \sum_\cokol M_\cokol \cdot H_\cokol.
$
%
Now apply the if direction of~\cref{lem:histogram} to get the histograms $H_M$ as required in condition~\ref{eq:2}.
In the other direction, assume condition~\ref{eq:2}.~holds,  
and use the only if direction of~\cref{lem:histogram} to decompose every $H_\cokol$ into simple
histograms. This yields
\[
N \ = \ \sum_\cokol M_\cokol \cdot S_\cokol,
\]
where all $M_\cokol \in \setM$ all $S_\cokol$ are simple histograms.
Finally we apply~\cref{lem:simplehistsum} to get matrices $N_\cokol$ satisfying~\eqref{eq:assumed}. This completes the proof.
\end{proof}

\para{Multihistograms}
%
Using Lemma~\ref{lem:histsum} we are now going to work out our final characterisation of the \zeroextprob, as formulated in Lemma~\ref{lem:first} below.
We write $H(i, \_)$ and $H(\_, j)$ for the $i$-th row and the $j$-th column of a matrix $H$, respectively.
For an indexed family $\{H_1, \ldots, H_k\}$ 
of matrices, 
its $j$-th column is defined as the indexed family of $j$-th columns
of respective matrices $\{H_1(\_, j), \ldots, H_k(\_, j)\}$.

Fix an input of the \zeroextprob: a matrix $\targetMatr$ and a finite set $\setM = \set{M_1, \ldots, M_k}$ of matrices, 
all of the same row dimension $d$.
Let $c_\cokol$ stand for the column dimension of $M_\cokol$. 
Suppose that some $N \in \zeroext{\targetMatr}$ and some family $\hist = \set{H_1, \ldots, H_k}$ of histograms
satisfy
%
\begin{align*}
N \ = \ M_1 \cdot H_1 \ + \ \ldots \ + \ M_k \cdot H_k.
\end{align*}
(The row dimension of every $H_\cokol$ is necessarily $c_\cokol$.)
Boiling down the equation to a single entry of $N$ we get a linear equation:
\begin{align*}
N(i, j) \ = \ & M_1(i, \_) \cdot H_1(\_, j) \ + \
                    \ldots \ + \
                   M_k(i, \_) \cdot H_k(\_, j).
\end{align*}
By grouping all the equations concerning all entries of a single column $N(\_, j) \in \Z^d$ of $N$ we get a system of $d$
(= row dimension of $N$) linear equations:
\begin{align*}
N(\_, j) \quad & = \quad M_1 \cdot H_1(\_, j) \ + \
                    \ldots \ + \
                   M_k \cdot H_k(\_, j)
                    \\
                   & = \quad  \Big[ M_1 \mid \ldots \ \mid M_k \Big] \cdot 
\begin{bmatrix}              
H_1(\_, j) \\
\ldots \\
H_k(\_, j)
\end{bmatrix}
\end{align*}
Therefore, the $j$-th column of $\hist$, treated as a single column vector of length $s = c_1 + \ldots + c_k$,
is a nonnegative-integer solution of a system of $d$ linear equations $\setU_{\setM, N(\_, j)}$, 
with $s$ unknowns $x_1 \ldots x_s$, of the form:
\begin{align*}
N(\_, j)  \quad = \quad 
                    \Big[ M_1 \mid \ldots \ \mid M_k \Big] \cdot 
\begin{bmatrix}              
x_1 \\
\ldots \\
x_s
\end{bmatrix},
\end{align*}
Observe that the system $\setU_{\setM, N(\_, j)}$ depends on $\setM$ and $N(\_, j)$ but not on $j$.
For succinctness, 
for $\vec a \in\Z^d$ we put 
\begin{align} \label{eq:C}
\setC_{\vec a} \quad := \quad \sol{\N}{\setU_{\setM, \vec a}}
\end{align}
to denote the set of all nonnegative-integer solutions of $\setU_{\setM, \vec a}$.
Therefore, every $j$th column of the multihistogram $\hist$ belongs to $\setC_{N(\_, j)}$.

Now recall that $N \in \zeroext{\targetMatr}$. Therefore, treating $\hist$ as a sequence of its column vectors in $\N^c$
(we call this sequence \emph{the word of $\hist$}), we 
arrive at the condition that this sequence belongs to the following language:
%
\begin{align} \label{eq:lang}
(\setC_{\zerovec})^* \ \setC_{\targetMatr(\_, 1)} \  
(\setC_{\zerovec})^* \ \setC_{\targetMatr(\_, 2)} \ \ldots \ 
(\setC_{\zerovec})^* \ \setC_{\targetMatr(\_, n)} \ 
(\setC_{\zerovec})^* 
\end{align}
where $n$ denotes the column dimension of $\targetMatr$.
If this is the case, we say that $\hist$ is an \emph{$(\targetMatr, \setM)$-multihistogram}. 
As the reasoning above is reversible, we have thus shown:

\begin{lemma} \label[lemma]{lem:first}
The \zeroextprob is equivalent to the following one:

\decproblem{\histprob}
{a finite set $\setM$ of matrices, and a matrix $\targetMatr$, all of the same row dimension $d$}
{does there exist an $(\targetMatr, \setM)$-multihistogram?}
%
\end{lemma}


\section{Upper bound for the \pproductprob} 
\label[section]{sec:upperbound}
%

We reduce in this section the \histprob (and hence also the \pproductprob, due to Lemmas~\ref{lem:equiv} and~\ref{lem:first}) 
to the VAS reachability problem (with single exponential blowup), thus obtaining 
decidability. 
Fix in this section an input to the \histprob: a matrix $\targetMatr$ (of column dimension $n$) 
and a finite set $\setM = \set{M_1, \ldots, M_k}$ of matrices, all of the same row dimension $d$.
We perform in two steps:
we start by proving an effective exponential bound on vectors appearing as columns of 
$(\targetMatr, \setM)$-multihistograms; 
then we construct a 
VAS whose runs correspond to the words of exponentially bounded $(\targetMatr, \setM)$-multihistograms.
For measuring the complexity we assume that all numbers in $\targetMatr$ and $\setM$ are encoded in binary.

\para{Exponentially bounded multihistograms} 
We need to recall first a characterisation of nonnegative-integer solution sets of systems of linear equations as
(effectively) exponentially bounded hybrid-linear sets,
i.e., of the form $B + P^\oplus$, for $B, P \subseteq \N^k$, where $k$ is the number of variables and
$P^\oplus$ stands for the set of all finite sums of vectors from $P$ (see e.g.~\cite{taming} (Prop.~2), \cite{Dom91}, \cite{Pot1991}).
By $\setU_{M, \vec a}$ denote a system of linear equations determined by a matrix $M$ and a column vector $\vec a$,
and by $\setU_{M,\zerovec}$ the corresponding \emph{homogeneous} systems of linear equations.  
Again, for measuring the size $\size{\setU_{M,\vec a}}$ of $\setU_{M,\vec a}$ we assume that all numbers in $M$ and $\vec a$ are encoded in binary.
%
\begin{lemma}[\cite{taming} Prop.~2] \label[lemma]{lem:taming}
$\sol{\N}{\setU_{M,\vec a}} = B + {P}^\oplus$, where $B, P \subseteq \N^k$ such that
all vectors in $B \cup P$ are bounded exponentially w.r.t.~$\size{\setU_{M,\vec a}}$ 
and $P \subseteq \sol{\N}{\setU_{M,\zerovec}}$. 
\end{lemma}
%

We will use Lemma~\ref{lem:taming} together with the following operation on multihistograms. 
A $j$-\emph{smear} of a histogram $H$ is any nonnegative matrix $H'$ obtained by replacing 
$j$-th column $H(\_, j)$ of $H$ by two columns that sum up to $H(\_, j)$. Here is an example ($j =5$):
\begin{align*}
\begin{bmatrix}
 3&	0&	0&	1&	\color{red}{0}&	0 & 0\\
 0&	1&	0&	0&	\color{red}{3}&	0 & 0 \\
  0&	0&	0&	1&	\color{red}{0}&	1 & 2
\end{bmatrix} \to
\begin{bmatrix}
 3&	0&	0&	1&	\color{red}{0}&	\color{blue}{0}&0 & 0\\
 0&	1&	0&	0&	\color{red}{2}&	\color{blue}{1}&0 & 0 \\
 0&	0&	0&	1&	\color{red}{0}&	\color{blue}{0}&1 & 2
\end{bmatrix}
\end{align*}
Formally, a $j$-smear of $H$ is any nonnegative matrix $H'$ satisfying:
\begin{align*}
 H'(\_, l) & \ =  \ H(\_, l) && \text{ for } l < j \\
 H'(\_, j) + H'(\_, j+1) & \ = \ H(\_, j) \\
 H'(\_, l+1) & \ = \ H(\_, l) && \text{ for } l > j. 
\end{align*}
One easily verifies that smear preserves the defining condition of histogram:
\begin{claim} \label{claim:smear}
A smear of a histogram is a histogram.
\end{claim}
\noindent
Finally, a $j$-smear of a family of matrices $\{H_1, \ldots, H_k\}$ is any 
indexed family of matrices $\{H'_1, \ldots, H'_k\}$
obtained by applying a $j$-smear simultaneously to all matrices $H_\cokol$. 
We omit the index $j$ when it is irrelevant.

So prepared, 
we claim that every $(\targetMatr, \setM)$-multihistogram $\hist = \set{H_1, \ldots, H_k}$
can be transformed by a number of smears into an $(\targetMatr, \setM)$-multihistogram 
containing only numbers exponentially bounded with respect to  $\targetMatr$, $\setM$.
Indeed, recall~\eqref{eq:lang} and let 
\[
N \ = \ \sum_{\cokol = 1, \ldots, k} M_\cokol \cdot H_\cokol \ \ \in \ \zeroext{\targetMatr}.
\]
Take an arbitrary (say $j$-th) column $\vec w \in \setC_{\vec a}$ of $\hist$ (recall~\eqref{eq:C}), where $\vec a = N(\_, j)$, 
treated as a single column vector $\vec w \in \N^s$
(for $s$ the sum of row dimensions of $H_1, \ldots, H_k$), and present it (using Lemma~\ref{lem:taming}) as a sum
\[
\vec w \quad = \quad \vec b \ + \ \vec p_1 \ + \ \ldots \ + \ \vec p_m,
\]
for some exponentially bounded $\vec b \in \setC_{\vec a}$ and $\vec p_1,  \ldots, \vec p_m \in \setC_{0}$.
Apply smear $m$ times, replacing the $j$-th column by $m+1$ columns $\vec b, \vec p_1, \ldots, \vec p_m$.
As $\vec b$ is a solution of the system $\setU_{\setM, \vec a}$ and every $\vec p_\cokol$ is a solution of the homogeneous system 
$\setU_{\setM, \zerovec}$, 
\begin{align*}
&\Big[ M_1 \mid \ldots \ \mid M_k \Big] \cdot \vec b            \quad\  = && \Big[ M_1 \mid \ldots \ \mid M_k \Big] \cdot \vec w \\
&\Big[ M_1 \mid \ldots \ \mid M_k \Big] \cdot \vec p_\cokol \quad = && \zerovec,
\end{align*}
the so obtained family $\hist' = \{H'_1, \ldots, H'_k\}$ still satisfies the condition
$\sum_{\cokol = 1, \ldots, k} M_\cokol \cdot H'_\cokol \in \zeroext{\targetMatr}$.
Using Claim~\ref{claim:smear} we deduce that $\hist'$ is an $(\targetMatr, \setM)$-multihistogram.
Repeating the same operation for every column of $\hist$ yields the required exponential bound.
%
%

\para{Construction of a VAS}\label[section]{sec:toVAS}
Given $\targetMatr$ and $\setM$ we now construct a VAS whose runs correspond to the words of
exponentially bounded $(\targetMatr, \setM)$-multihistograms. 
Think of the VAS as reading (or nondeterministically guessing) 
consecutive column vectors (i.e., the word) of a potential $(\targetMatr, \setM)$-multihistogram 
$\hist = \{H_1, \ldots, H_k\}$.
The VAS has to check two conditions: 

\begin{enumerate}
\item the word of $\hist$ belongs to the language~\eqref{eq:lang}; 
\item the matrices $H_1, \ldots, H_k$ satisfy the histogram condition.
\end{enumerate}

\noindent
The first condition, under the exponential bound proved above, amounts to the membership in a regular language and 
can be imposed by a VAS in a standard way.
The second condition is a conjunction of $k$ histogram conditions, and again the conjunction can be realised in a standard way. 
We thus focus, from now on, only on showing that a VAS can check that its input is a histogram.

To this aim it will be profitable to have the following characterisation of histograms.
For an arbitrary $\matr{r}{c}$-matrix $H$, define the $\matr{(r-1)}{c}$-matrix $\prof{H}$ as:
\begin{align*}
\prof{H}(i,j) \quad \eqdef \quad \sum H_0(i, 1 \ldots j) - \sum H_0(i+1, 1\ldots j+1),
\end{align*}
where $H_0$ is an $\matr{r}{(c+1)}$-matrix which extends $H$ by the $(c+1)$-th zero column.
\begin{lemma}\label[lemma]{lem:profiles}
A nonnegative $\matr{r}{c}$-matrix $H$ is a histogram if, and only if
$\prof{H}$ is nonnegative and $\prof{H}(\_, c) = \zerovec$.
\end{lemma}
\begin{proof}
Indeed, nonnegativeness of $\prof{H}$ is equivalent to saying that 
$$
\sum H(i, 1 \ldots j) \ \geq \ \sum H(i+1, 1\ldots j+1)
$$
for every $1\leq i<r$ and $0\leq j < c$;
moreover, $\prof{H}(\_, c) = \zerovec$ is equivalent to saying that 
$\sum H(i, 1 \ldots c)$ is the same for every $i = 1,\ldots, r$.
\end{proof}

For technical convenience, we always extend $\prof{H}$ with an additional very first zero column $\zerovec$; in other words,
we put $\prof{H}(\_, 0) = \zerovec$. Here is a formula relating two consecutive columns $\prof{H}(\_, j-1)$ and
$\prof{H}(\_, j)$ of $\prof{H}$ and two consecutive columns $H(\_, j)$ and $H(\_, j+1)$ of $H$,
\begin{align} \label{eq:inv}
\prof{H}(i, j) \ = \  \prof{H}(i, j-1) + H(i, j) - H(i+1, j+1),
\end{align}
that will lead our construction.

We now define a VAS of dimension $2(r-1)$ that reads consecutive columns $\vec w \in \N^r$ of an exponentially
bounded matrix and accepts if, and only if the matrix is a histogram.
According to the convention that $\prof{H}(\_, 0) = \zerovec$, all the $2(r-1)$ counters are initially set to 0.
Counters $1, \ldots, r-1$ of the VAS are used as a buffer to temporarily store the input;
counters $r, \ldots, 2(r-1)$ ultimately store the current column $\prof{H}(\_, j)$.
According to~\eqref{eq:inv}, the VAS obeys the following invariant: after $j$ steps,
\begin{align} \label{eq:realinv}
\prof{H}(i, j) \ = \ \text{counter}_i + \text{counter}_{r-1+i}.
\end{align}
Let $\setC \subseteq \N^r$ denote the exponential set of all column vectors that can appear in a histogram, as derived above.
For every $\vec w = (w_1, \ldots, w_r) \in \setC$, the VAS has a 'reading' transition
that adds $(w_1, \ldots, w_{r-1}) \in \N^{r-1}$ to its counters $1, \ldots, r-1$, and subtracts $(w_2, \ldots, w_r) \in \N^{r-1}$ from 
its counters $r, \ldots, 2(r-1)$ (think of $\vec w(i+1) = H(i+1, j+1)$ in the equation~\eqref{eq:inv}).
Furthermore, for every $i = 1, \ldots, r-1$ the VAS has a 'moving' 
transition that subtracts $1$ from counter $i$ and adds $1$ to counter
$r-1+i$, i.e., moves $1$ from counter $i$ to counter $r-1+i$.
(recall the '$+H(i, j)$' summand in the equation~\eqref{eq:inv}).
Observe that these transitions preserve the invariant~\eqref{eq:realinv}.

Relying on Lemma~\ref{lem:profiles} 
we claim that the so defined VAS reaches nontrivially (i.e., along a nonempty run) 
the zero configuration (all counters equal 0) if, and only if its input, treated
as an $\matr{r}{c}$-matrix $H$, is a histogram with all entries belonging to $\setC$.
In one direction, the invariant~\eqref{eq:realinv} assures that $\prof{H}$ is nonnegative and
the final zero configuration assures that $\prof{H}(\_, c) = \zerovec$.
In the opposite direction, if a histogram is input, the VAS has a run ending in the zero configuration.
The VAS is computable in exponential time (as the set $\setC$ above is so).

We have shown that, given $\targetMatr$ and $\setM$, one can effectively built a VAS which
admits reachability if, and only if there exists an $(\targetMatr, \setM)$-multihistogram.
The (exponential-blowup) reduction of the \pproductprob 
to the VAS reachability problem is thus completed.


\section{\ptime decision procedures} \label[section]{sec:p}

In this section we prove Theorem~\ref{thm:p}, namely we provide polynomial-time decision procedures
for the \xpproductprob{$\X$}, where $\X \in \{\Z, \Q, \Qplus\}$.
The most interesting case $\X = \Qplus$ is treated in Section~\ref{sec:Qplus}.
The remaining ones are in fact special cases of a more general result, shown in Section~\ref{sec:pp},
that applies to an arbitrary commutative ring.

\subsection{
$\X = \Qplus$} \label{sec:Qplus}


We start by noticing that the whole development of (multi)-his\-to\-grams in Section~\ref{sec:hist} is not at all specific
for $\X=\N$ and works equally well for $\X = \Qplus$.
It is enough to relax the definition of histogram: instead of nonnegative integer matrix, let histogram be now a
\emph{nonnegative rational} matrix satisfying exactly the same conditions as in Definition~\ref{def:hist} in Section~\ref{sec:hist}.
In particular, the degree of a histogram is now a nonnegative rational.
Accordingly, one adapts the \zeroextprob and considers a sum of 0-extensions of $\setM$ \emph{multiplied by nonnegative rationals}.
The same relaxation as for histograms we apply to multihistograms, and in the definition of the latter (cf.~the language~\eqref{eq:lang}
at the end of Section~\ref{sec:hist}) we consider
nonnegative-rational solutions of linear equations instead of nonnegative-integer ones.
With these adaptations, the \xpproductprob{$\Qplus$} is equivalent to the following decision problem
(whenever a risk of confusion arises, we specify explicitly which matrices are integer ones, and which rational ones):

\decproblem{\xhistprob{$\Qplus$}}
{a finite set $\setM$ of integer matrices, and an integer matrix $\targetMatr$, all of the same row dimension $d$}
{does there exist a rational $(\targetMatr, \setM)$-multihistogram?}

\noindent
From now on we concentrate on the polynomial-time decision procedure for this problem.
We proceed in two steps. 
First, we define \emph{homogeneous linear Petri nets}, a variant of
Petri nets generalising continuous PNs~\cite{serge}, and show how to solve its reachability problem 
by $\Qplus$-solvability of a slight generalisation of linear equations (linear equations with implications),
following the approach of~\cite{serge-compl}.
Next, using a similar construction as in \cref{sec:upperbound}, combined with the above characterisation of 
reachability, we encode \xhistprob{$\Qplus$} as a system of linear equations with implications.

\para{Homogeneous linear Petri nets}
%
%
%
A \emph{homogeneous linear Petri net} (homogeneous linear PN) of dimension $d$ is a finite set of homogeneous\footnote{
If non-homogeneous systems were allowed, the model would subsume (ordinary) Petri nets.
}
 systems of linear equations 
$\V = \{\setU_1, \ldots, \setU_m\}$, called \emph{transition rules}, all over the same $2d$ variables $x_1, \ldots, x_{2d}$.
The transition rules determine a transition relation $\longrightarrow$ between configurations, which are nonnegative rational vectors 
$\vec{c}\in (\Qplus)^\dimension$, as follows:
there is a transition $\vec{c} \longrightarrow \vec{c}'$ if, for some $i \in \{1, \ldots, m\}$ and 
$\vec v \in \sol{\Qplus}{\setU_i}$,
the vector $\vec c \ - \ \pi_{1\ldots d}(\vec v)$ is still a configuration, and
$$\vec{c}' \ = \ \vec{c} - \pi_{1\ldots d}(\vec v) \ + \ \pi_{d+1\ldots 2d}(\vec v).$$
(The vectors $ \pi_{1\ldots d}(\vec v)$ and $\pi_{d+1\ldots 2d}(\vec v)$ are projections of $\vec v$ on respective coordinates.)
The reachability relation
$\vec{c} \longrightarrow^* \vec{c'}$ holds, if there is a sequence of transitions (called \emph{a run}) from $\vec c$ to $\vec c'$.


A class of \emph{continuous PN} \cite{serge} is a subclass of homogeneous linear PN, where every
system of linear equations $\setU_i$ has a 1-di\-men\-sio\-nal solution set of the form
$\setof{c\vec{v}}{c\in\Qplus}$, for some fixed $\vec v\in \N^{2d}$.

\para{Linear equations with implications}
A \emph{\semieq} is a finite set of linear equations, all over the same variables, 
plus a finite set of implications of the form
\[
x > 0 \implies y > 0,
\]
where $x, y$ are variables appearing in the linear equations.
The solutions of a \semieq are defined as usually, but additionally they must satisfy all implications.
The $\Qplus$-solvability problem asks if there is a nonnegative-rational solution.
In~\cite{serge-compl} (Algorithm 2) it has been shown (within a different 
notation) how to solve the problem in \ptime; another proof is derivable from~\cite{BlondinH17}, where a polynomial-time
fragment of existential FO($\Q$, + ,<) has been identified that captures \semieq:
\begin{lemma}[\cite{serge-compl,BlondinH17}]\label{lem:serge-compl}
The $\Qplus$-solvability problem for \semieqs is decidable in \ptime.
\end{lemma}

\noindent
Due to~\cite{serge-compl},
the reachability problem for continuous PNs reduces to the $\Qplus$-solvability of \semieqs.
We generalise this result and prove the reachability relation of a homogeneous linear PN to be 
effectively described by a \semieq:
\begin{lemma} \label{lem:LPN}
Given a homogeneous linear PN $\V$ of dimension $d$ (with numbers encoded in binary) one can compute in 
polynomial time a \semieq whose 
$\Qplus$-solution set, projected onto a subset of $2d$ variables, describes the reachability relation of $\V$.
\end{lemma}
\noindent
We return to the proof of this lemma, once we first use it in the decision procedure for our problem.

\para{Polynomial-time decision procedure}
Now, we are ready to describe a decision procedure for the
\xhistprob{$\Qplus$}, by a polynomial-time reduction to the 
$\Qplus$-solvability problem of \semieqs.

Fix an input to the \xhistprob{$\Qplus$}, i.e., $\targetMatr$ and $\setM = \{M_1, \ldots, M_k\}$.
Analogously as in~\eqref{eq:C} in Section~\ref{sec:hist} we put for succinctness, for $\vec a\in\Z^d$,
\[
\setC_{\vec a} \quad := \quad \sol{\Qplus}{\setU_{\setM, \vec a}} \quad  \subseteq \quad (\Qplus)^r 
\]
to denote the set of all nonnegative-rational solutions of the system $\setU_{\setM, \vec a}$ of linear equations determined by the matrix
\[
\Big[ M_1 \mid \ldots \ \mid M_k \Big] 
\]
and the column vector $\vec a$. Recall the language~\eqref{eq:lang}:
\begin{align} \label{eq:langQplus}
(\setC_{\zerovec})^* \ \setC_{\targetMatr(\_, 1)} \  
(\setC_{\zerovec})^* \ \setC_{\targetMatr(\_, 2)} \ \ldots \ 
(\setC_{\zerovec})^* \ \setC_{\targetMatr(\_, n)} \ 
(\setC_{\zerovec})^*,
\end{align}
where $n$ is the column dimension of $\targetMatr$.
Our aim is to check existence of an $(\targetMatr, \setM)$-multihistogram, i.e., of a family 
$\hist = \{H_1, \ldots, H_k\}$ of nonnegative-rational matrices, such that the following conditions are satisfied:
\begin{enumerate}
\item the word of $\hist$ belongs to the language~\eqref{eq:langQplus}; 
\item the matrices $H_1, \ldots, H_k$ satisfy the histogram condition.
\end{enumerate}
\noindent $\hist$ has 
$r = r_1 + \ldots + r_k$ rows, where $r_\cokol$ is the row dimension of $H_\cokol$, equal to the column
dimension of $M_\cokol$ for $\cokol = 1,\ldots, k$.
Given $\targetMatr, \setM$, we construct in polynomial time a \semieq $\S$ which is solvable if, and only if
some $(\targetMatr, \setM)$-multihistogram exists.
The solvability of $\S$ is decidable in \ptime according to~\cref{lem:serge-compl}.

The idea is to characterise an $(\targetMatr, \setM)$-multihistogram by a sequence of runs in a homogeneous linear PN interleaved by
single steps described by non-homogeneous systems of linear equations;
then, using \cref{lem:LPN}, 
the sequence is translated to a \semieq.
Conceptually, the construction is analogous to the construction of a VAS in Section~\ref{sec:upperbound}.
We define a homogeneous linear PN $\V_\zerovec$, recognizing the language $(\setC_\zerovec)^*$.
(Think of $\V_\zerovec$ as if it reads consecutive nonnegative-rational column vectors belonging to $\setC_\zerovec$.)
The dimension of $\V_\zerovec$ is 
$$2(r-k) \ = \ 2(r_1 - 1) + \ldots + 2(r_k - 1),$$ i.e., $\V_\zerovec$ has $2(r_i-1)$ counters corresponding to each $H_i$.
%

Concerning the 'reading' transitions,
the construction is essentially the same, except that we do not restrict the input to the finite set
$\setC$ as in Section~\ref{sec:upperbound}, but we allow for all infinitely many 
solutions $\setC_\zerovec$ of the homogeneous system of linear equations $\setU_{\setM, \zerovec}$
as inputs.
Moreover, we deal with all histogram conditions for $H_1, \ldots, H_k$ simultaneously.
Thus, the 'reading' transition rule of $\V_\zerovec$ is described by 
a homogeneous system $\setU_\zerovec$ over $4(r-k)$ variables, 
half of them describing subtraction and half describing addition in a transition,
derived from $\setU_{\setM, \zerovec}$ as follows.
Let $x_{1}, \ldots, x_{r_j}$ be the variables of $\setU_{\setM, \zerovec}$ corresponding to rows of some $H_j$.
As $\V_\zerovec$ has $2(r_j-1)$ dimensions corresponding to $H_j$, 
the system $\setU_\zerovec$ has $4(r_j - 1)$ corresponding variables $z_1, \ldots, z_{4(r_j-1)}$,  
half of them, say $z_1, \ldots, z_{2(r_j-1)}$, describing subtraction and the other half $z_{2(r_j -1) +1}, \ldots, z_{4(r_j-1)}$ addition.
Imitating the construction of a VAS in Section~\ref{sec:upperbound},
the system $\setU_\zerovec$ is obtained by adding to $\setU_{\setM, \zerovec}$, for every $j$, the following equations:
\begin{align*}
\big(z_{2(r_j-1)+1}, \ldots, z_{3(r_j-1)}\big) \ & = \ \big( x_1, \ldots, x_{r_j-1} \big) 
\end{align*}
describing addition in dimensions $1, \ldots, r_j-1$; and the following equations:
\begin{align*}
\big(z_{(r_j-1)+1}, \ldots, z_{2(r_j-1)}\big) \ & = \ \big( x_2, \ldots, x_{r_j} \big) 
\end{align*}
describing subtraction in dimensions $(r_j-1)+1, \ldots, 2(r_j-1)$;
and then by eliminating the (redundant) variables $x_1, \ldots, x_{r_j}$.
%
%

Concerning the 'moving' transitions,
%
%
there are $r-k$ of them in $\V_\zerovec$,
each one described by a separate system of homogeneous linear equations $\setW_i$, 
consisting of just one equation of the form (using the same indexing as above)
\[
z_{\cokol} \ = \ z_{3(r_j-1) + \cokol}, \qquad \text{ where } 1\leq \cokol \leq r_j-1, \ 1 \leq j \leq k.
\]
In addition, both in $\setU_\zerovec$ and in  $\setW_i$, all variables $z_\cokol$ not mentioned above are equalised to $0$.

Summing up, $\V_\zerovec \eqdef \{\setU_{\zerovec},
\setW_1, \ldots, \setW_{r-k}\}$.
 By Lemma~\ref{lem:LPN} one can compute in polynomial time a \semieq $\S_\zerovec$ 
 such that the projection $A_\zerovec$ of $\sol{\Qplus}{\S_\zerovec}$ to some 
$2\cdot 2(r-k)$ of its variables describes the reachability relation
 of $\V_\zerovec$.
 
According to~\eqref{eq:langQplus}, we aim at constructing a \semieq $\S$ whose solvability is equivalent to
existence of the following sequence of $n+1$ runs  of $\V_\zerovec$:
\begin{align} \label{eq:langQplus1}
\zerovec \xrightarrow{\setC_{\zerovec}^*} \vec c_1 \quad 
\vec c_2 \xrightarrow{\setC_{\zerovec}^*}\vec c_3 \ 
\ldots \ 
\vec c_{2n-2} \xrightarrow{\setC_\zerovec^*}
\vec c_{2n-1} \quad 
\vec c_{2n} 
\xrightarrow{\setC_{\zerovec}^*}\zerovec
\end{align}
where the relation between every ending configuration $\vec c_{2i-1}$ and every next starting configuration $\vec c_{2i}$
is determined by $\setC_{\targetMatr(\_, i)}$.
The required \semieq $\S$ is constructed as follows: we introduce $2(r-k)$ variables per each intermediate configuration
$\vec c_i$ ($\vec c_0=\vec c_{2n+1}=\zerovec$), and impose the constraints:
\begin{enumerate}
\item there is a run from $\vec c_{2i}$ to $\vec c_{2i+1}$ in $\V_{\zerovec}$, i.e., $(\vec c_{2i}, \vec c_{2i+1})$ 
belongs to the projection $A_\zerovec$ of $\sol{\Qplus}{\S_\zerovec}$;
\item $\vec c_{2i}-\vec c_{2i-1}\in \sol{\Qplus}{\setU_i}$, where the (nonhomogeneous) system $\setU_i$ is obtained 
from  $\setU_{\setM, \targetMatr(\_, i)}$ similarly as $\setU_\zerovec$ above.
\end{enumerate}
$\S$ is solvable iff some $(\targetMatr, \setM)$-multihistogram exists.
 
\medskip

\para{Proof of Lemma~\ref{lem:LPN}}
We start by observing that the reachability of a homogeneous PN can be 
simulated by a continuous PN:
\begin{lemma}\label{lem:conEqlin}
For a given homogeneous PN $\cal V$ one can construct a continuous PN 
$\calN$such 
that $\vec{c}\to \vec{c'}$ in $\cal V$ iff $\vec{c}\to^* \vec{c'}$ in $\calN$. 
\end{lemma}
\noindent
Thus, knowing that homogeneous linear PN subsume continuous PN, the former are 
potentially (exponentially) more succinct representations of the latter.

We need to recall a crucial observation on continuous PN, made in Lemma 12 in~\cite{serge-compl}, 
namely that whenever the initial and final configuration have positive values on all coordinates used
by the transitions in a run, the reachability reduces to $\Qplus$-solvability of state equation.  
Using \cref{lem:conEqlin}, we translate this observation to homogeneous linear 
PN.
Configurations below are understood to be elements of $(\Qplus)^d$, where $d$ is 
the dimension.
Recall that a homogeneous linear PN is determined by a finite set of 
homogeneous systems of linear equations $\setU$;
and that its transition is determined by a solution $\vec v \in \sol{\Qplus}{\setU}$ of some of the systems,
namely the transition 
first subtracts
$\vec v^- = \pi_{1\ldots \dimension} (\vec{v})$ from the current configuration,
and then adds $\vec v^+ = \pi_{\dimension+1\ldots 2\dimension} (\vec{v})$ to it.

%

\begin{lemma}\label[lemma]{lem:serge12lin}
Let $\V = \{\setU_1\ldots \setU_k\}$ be a $d$-dimensional homogeneous linear PN. 
A configuration $\vec{f}$ is reachable from a configuration $\vec{i}$ in $\V$ whenever 
there are some solutions $\vec{u}_i \in \sol{\Qplus}{\setU_i}$, for $i = 1\ldots, k$, satisfying the following conditions:
 \begin{enumerate}
  \item $\vec{f}-\vec{i} = \sum_{i=1}^{k} - \vec u_i^- + \vec u_i^+$;
  \item if $\vec u_i^-(j) > 0$ then $\vec{i}(j)>0$;
  \item if $\vec u_i^+(j)>0$ then $\vec{f}(j)>0$;
 \end{enumerate}
with $i$ ranging over $1\ldots k$ and $j$ over $1\ldots d$.
\end{lemma}
%
%
%
%
%
\noindent
Thus \cref{lem:serge12lin} provides a sufficient (but not necessary) condition for $\vec i \longrightarrow^*\vec f$.
We now use the lemma to fully characterise the reachability relation in 
homogeneous linear PN:

\begin{lemma}\label{lem:serge20lin}
Let $\V = \{\setU_1\ldots \setU_k\}$ be $\dimension$-dimensio\-nal homogeneous 
linear PN.
A configuration $\vec{f}$ is reachable from a configuration $\vec{i}$ if, and only if 
for some two configurations $\vec{i'}$ and $\vec{f'}$:
 \begin{enumerate}
  \item $\vec{i'}$ is reachable from $\vec{i}$ in at most $\dimension$ steps;
  \item $\vec{f}$ is reachable from $\vec{f'}$ in at most $\dimension$ steps;
  \item $\vec{i'}$ and $\vec{f'}$ satisfy the sufficient condition of \cref{lem:serge12lin}.
  \end{enumerate}
\end{lemma}
\begin{proof}
The if direction is immediate. For the only if direction, 
%
%
%
consider a fixed run from $\vec i$ to $\vec f$, i.e., a sequece of $\Qplus$-solutions of systems
$\setU_1$, \ldots, $\setU_k$. For every $i = 1,\ldots, m$, let $\vec u_i$ be the sum of the multiset of $\Qplus$-solutions of $\setU_i$ 
that take part in the run; clearly $\vec{u}_i \in \sol{\Qplus}{\setU_i}$ as the solution set is additive.

As the only requirement we demand for $\vec{i}'$ (and $\vec f'$) is its positivity on certain coordinates, we modify the run by
requiring that  
in its few first steps it executes transitions with shrinking quantities 
guarantying that the number of coordinates 
equal $0$ in the intermediate configurations is monotonically non-increasing. 
Furthermore, we may also require that every of the few first steps increases the number of 
nonzero coordinates. 
Thus, after at most $d$ steps a configuration $\vec i'$ is reached that achieves the maximum number of 
nonzero coordinates.
Likewise, reasoning backward, we obtain a configuration $\vec f'$ with the same property.
Thus $\vec i'$ and $\vec f'$ satisfy the conditions 2.~and 3.~of \cref{lem:serge12lin}.
Moreover, we may assume that the sum of solutions $\vec v_i  \in \sol{\Qplus}{\setU_i}$ that take part in the runs
$\vec i \longrightarrow^* \vec i'$ and $\vec f' \longrightarrow^* \vec f$
satisfy $\vec v_i \leq \vec u_i$, for $i = 1, \ldots, k$, as arbitrary small quantities are sufficient for positiveness of $\vec i'$ and $\vec f'$. 
The inequalities allow us to derive condition 1.~of \cref{lem:serge12lin}, as
$\vec u_i - \vec v_i \in \Qplus$ for $i =1, \ldots, k$.
Indeed, the run $\vec i  \longrightarrow^* \vec f$ implies the first equality below,
and the two runs 
$\vec i \longrightarrow^* \vec i'$ and $\vec f' \longrightarrow^* \vec f$ imply the second one:
$$\vec{f}-\vec{i} = \sum_{i=1}^{k} - \vec u_i^- + \vec u_i^+ \qquad\qquad
(\vec{i}'-\vec{i}) + (\vec f - \vec f') = \sum_{i=1}^{k} - \vec v_i^- + \vec v_i^+.$$ 
Subtraction of the two equalities yields:
$$\vec{f}-\vec{i} = \sum_{i=1}^{k} - (\vec u_i - \vec v_i)^- + (\vec u_i - \vec v_i)^+,$$ 
as required in condition 1.~of \cref{lem:serge12lin}.
%
\end{proof}
%

Now we are prepared to complete the proof of Lemma~\ref{lem:LPN}.
The \semieq $\S$ is constructed relying directly on the characterisation of \cref{lem:serge20lin}.
Linear equations are used to express the two runs of length bounded by $d$, 
as well as the condition 1.~of \cref{lem:serge12lin} (which appears in condition 3.~in \cref{lem:serge20lin}); and 
implications are used to express conditions 2.~and 3.~of \cref{lem:serge12lin}. 
The size of $\S$ is clearly polynomial in $\V$.

\subsection{
$\X \in \{\Z, \Q\}$}  \label{sec:pp}

In this, and only in this section we generalise slightly our setting and
consider a fixed commutative ring $\R$, instead of just the ring $\Q$ or rationals.
Accordingly, by a data vector we mean in this section a function
$\D \to \R^d$ from data values to $d$-tuples of elements of $\R$ that map almost all
data values to the vector of zero elements $\zerovec\in\R^d$.
With this more general notion of data vectors, we define
\xpproducts{$\R$} and the \xpproductprob{$\R$} analogously as in Section~\ref{sec:datavectors}.
Furthermore, we define analogously $\R$-sums and consider linear equations with coefficients from $\R$ and their $\R$-solvability problem.

\begin{theorem}\label[theorem] {thm:pp}
For any commutative ring $\R$, the \xpproductprob{$\R$} reduces polynomially
to the $\R$-solvability problem of linear equations.
\end{theorem}

\noindent
Clearly, Theorem~\ref{thm:pp} implies the remaining cases of Theorem~\ref{thm:p}, namely $\X \in \{\Z, \Q\}$,
as in these cases the $\X$-solvability of linear equations is in \ptime.
\Cref{thm:pp} follows immediately by \Cref{lem:THM3}, stated below, whose proof is strongly inspired by Theorem 15 in~\cite{HLT2017LICS}.
For a data vector $\vec v$, we define the vector $\undata{\vec v}\in \R^d$ and 
a finite set of vectors $\compon{\vec v} \finsubseteq \R^d$:
\begin{align*}
\undata{\vec v} \quad & \eqdef \quad \sum_{\alpha \in \support{\vec v}} \vec v(\alpha) \\
\compon{\vec v} \quad & \eqdef \quad \setof{\vec v(\alpha)}{\alpha\in\support{\vec v}}.
\end{align*}
Clearly both operations commute with data permutations: $\undata{\vec v} = \undata{\vec v \circ \theta}$
and $\compon{\vec v} = \compon{\vec v \circ \theta}$, and can be 
lifted naturally to finite sets of data vectors:
\begin{align*}
\undata{\V} \quad & \eqdef \quad \setof{\undata{\vec v}}{\vec v \in \V} \\
\compon{\V} \quad & \eqdef \quad \bigcup_{\vec v \in \V} \compon{\vec v}.
\end{align*}
%
%
%
%
%
\begin{lemma}\label{lem:THM3} 
Let $\target$ be a data vector and $\setV$ be a finite set of data vectors 
$\setV$. 
Then $\target$ is an \xpproduct{$\R$} of $\setV$ 
if, and only if 
\begin{enumerate}
 \item $\undata{\target}$ is an $\R$-sum of $\undata{\setV}$, and 
 \item every $\vec a \in \compon{\target}$ is an $\R$-sum of $\compon{\setV}$.
\end{enumerate}
\end{lemma}
\begin{proof}
The only if direction is immediate: if $\target = z_1 \cdot \vec w_1 + \ldots + z_n \cdot \vec w_n$ for $z_1, \ldots, z_n \in \R$ and
$\vec w_1, \ldots, \vec w_n \in \perm{\setV}$, then clearly 
$\undata{\target} = z_1 \cdot \undata{\vec w_1} + \ldots + z_n \cdot \undata{\vec w_n}$
and hence $\undata{\target}$ is a $\R$-sum of $\undata{\setV}$ (using the fact that $\undata{\_}$ commutes with data permutations).
Also $\target(\alpha)$ is necessarily an $\R$-sum of $\compon{\setV}$ for every $\alpha \in \support{\target}$.

For a vector $\vec a \in \R^d$, we define an \emph{$\vec a$-move} as an arbitrary data vector that 
maps some data value to $\vec a$, some other data value to $- \vec a$, and all other data values to $\zerovec$.
\begin{claim}\label{claim:move}
Every $\vec a$-move, for $\vec a \in \compon{\vec v}$, is an \xpproduct{$\R$} of $\{\vec v\}$.
\end{claim}
%
\noindent
Indeed, for $\vec a = \vec v(\alpha)$, 
consider a data permutation $\theta$ which preserves all elements of $\support{\vec v}$ except that it maps $\alpha$
to a data value $\alpha'$ related in the same way as $\alpha$ by the order $\leq$ to other data values in $\support{\vec v}$.
Then ${\vec a}$-moves are exactly data vectors 
$(\vec v - \vec v \circ \theta) \circ \rho = \vec v \circ \rho - \vec v \circ (\theta \circ \rho)$.

For the if direction, suppose point 1. holds: 
$\undata{\target}$ is an $\R$-sum of $\undata{\setV}$.
Treat the vector $\undata{\target}$ and the vectors in $\undata{\setV}$ as data 
vectors with the same singleton support.
Observe that $\undata{\vec{v}}$ for any $\vec{v}\in\setV$ is  
an \xpproduct{$\R$} of $\{\vec{v}\}$; indeed we can use $\vec{a}$-moves to 
transfer all nonzero vectors for data in $\support{\vec{v}}$ into one datum.
With this view in mind
we have:
\begin{itemize}
 \item $\undata{\target}$ is an \xpproduct{$\R$} of $\setV.$
\end{itemize}
%

\noindent
Furthermore, suppose point 2. holds: every $\vec a \in \compon{\target}$ is an $\R$-sum of $\compon{\setV}$.
Thus every $\vec{a}$-move, for $\vec a \in \compon{\target}$,  
is an $\R$-sum of $\setof{\vec{b}\text{-move}}{\vec{b}\in \compon{\setV}}.$
By Claim~\ref{claim:move} we know that every element of the latter set is an \xpproduct{$\R$} of $\setV$. Thus we entail:
\begin{itemize}
\item every $\vec a$-move, for $\vec a \in \compon{\target}$, is an \xpproduct{$\R$} of $\setV$.
\end{itemize}
%
%

\noindent
We have shown that $\undata{\target}$, as well as all $\vec a$-moves (for all $\vec a \in \compon{\target}$),
are \xpproducts{$\R$} of $\setV$.
We use the $\vec{a}$-moves to transform $\undata{\target}$ into $\target$, which proves that 
$\target$ is an \xpproduct{$\R$} of $\setV$ as required.
\end{proof}


\section{Concluding remarks} \label{sec:orbitfinite}

The main result of this paper is determining the computational complexity of solving linear equations with integer
(or rational) coefficients, over ordered data.
We observed the huge gap: while the $\N$-solvability problem is equivalent (up to an exponential blowup) to the VAS reachability problem,
the $\Z$-, $\Q$-, and $\Qplus$-solvability problems are all in \ptime.

Except for the last Section~\ref{sec:pp}, we assumed in this paper that the coefficients and solutions come from
the ring $\Q$ of rationals, but clearly one can consider other commutative rings as well.
There is another possible axis of generalisation, which we want to mention now, namely
orbit-finite systems of linear equations over an orbit-finite set of variables.
%
%

Fix a commutative ring $\R$.
Let $\calX, \calY$ be arbitrary, possibly infinite sets.
By an \emph{$\cal X$-vector} we mean any function
$
\calX \to \R
$
which maps almost all elements of $\calX$ to $0\in\R$. 
An $\matr{\calX}{\calY}$-matrix is an $\calX$-indexed family of (column) $\calY$-vectors,
\[
M \quad = \quad (M_x \in \R^\calY)_{x\in\calX}.
\]
Such a matrix $M$, together with a (column) $\calY$-vector $\vec a$, determines a system of linear equations
$\setU_{M,\vec a}$, whose solutions are those $\calX$-vectors which, treated as coefficients of 
a linear combination of vectors $M_x$, yield $\vec a \in \calY$:
\[
\solsing{\setU_{M,\vec a}} \quad = \quad 
\setof{\vec v\in\R^\calX}{\sum_{x\in\calX} \vec v(x)\cdot M_x = \vec a}.
\]
Note that the sum is well defined as $\vec v(x) \neq 0$ for only finitely many elements $x\in \calX$.
The setting of this paper is nothing but a special case, where $\R = \Q$; and where 
\[
\calX = \perm{\setV} \quad \text{ and } \quad \calY = \D \times \{1, \ldots, d\}
\]
are \emph{orbit-finite} sets,
i.e., sets which are finite up to the natural action of automorphisms of the data domain $(\D, \leq)$:
data vectors are clearly elements of $\Z^{\D \times \{1, \ldots, d\}}$, and solutions we seek for 
are essentially elements of $\R^{\perm{\setV}}$.
The natural action of a monotonic bijection $\theta : \D \to \D$ maps a pair $(d, i)\in\D \times \{1, \ldots, d\}$ 
to $(\theta(d), i)$; and maps a data vector $\vec v\in \setV$ to $\vec v \circ \theta^{-1}$.
Similarly, another special case has been investigated in~\cite{HLT2017LICS}, where finiteness
up to the natural action of automorphisms of the data domain $(\D, =)$ played a similar role.
As another example, in~\cite{locfin} the orbit-finite solvability problem has been investigated (in the framework of CSP)
for the same data domain $(\D, =)$, in the case where $\R$ is a finite 
field.

It is an exciting research challenge to fully understand the complexity landscape of orbit-finite systems of linear equations, 
as a function of the choice of data domain.
In this direction, the results of this paper are discouraging: the case of ordered data, 
compared to the case of unordered data investigated in~\cite{HLT2017LICS}, 
requires significantly new techniques and the complexity of the nonnegative integer solvability differs significantly too;
thus it is expectable that different choices of data domain will require different approaches.
Nevertheless, investigation of orbit-finite dimensional linear algebra seems to be a tempting continuation of our work.

\bibliographystyle{plain}

\appendix

\newpage


\section{Missing proofs}

\medskip

\begin{proof}[Proof of \cref{lem:histogram}]
The if direction is easy as sum of histograms is a histogram, and the degree is the sum of degrees.

We prove the only if direction by induction on the degree $s$ of a histogram. 
For $s=1$ the claim is trivial, so assume $s>1$ and fix a $\matr{r}{c}$-histogram $H$ of degree $s$.
We are going to extract from $H$ a simple histogram $S$ 
in such a way that remaining matrix $H - S$ is still a histogram, necessarily of degree $s-1$. 
Then one can use the induction assumption to deduce that $H$ 
can be decomposed as a sum of simple histograms.

Consider a function a function $f: \set{1\ldots r}\to\set{1\ldots c}$ which maps $i$ to the smallest $j$ with $H(i,j)>0$.
Let $S$ be the simple $\matr{r}{c}$-histogram induces by $f$.
We need to check that the matrix $H-S$ is a histogram.
As the first defining condition of histogram is obvious, we concentrate on the second one, i.e., 
for any $1\leq i<r$ and $0\leq j < c$ we aim at showing 
\[
\sum \big[H-S\big](i, 1 \ldots j) \geq \sum \big[H-S\big](i+1, 1\ldots j+1).
\]
We consider separately three cases:

\begin{itemize}
\item[]{\it (i)} \ $f(i)<f(i+1) \leq j+1$;
\item[]{\it (ii)} \  $f(i)\leq j+1 <f(i+1)$; and
\item[]{\it (iii)} \ $j+1<f(i)<f(i+1)$.
\end{itemize}

\noindent
In the case {\it (i)} we have 
\begin{align*}
\sum \big[H-S\big](i,1\ldots j)  \ = & \ \sum H(i,1\ldots j) - 1 \geq \\ 
\sum H(i+1,1\ldots j+1) - 1 \ = & \ \sum  \big[H-S\big](i+1,1\ldots j+1);
\end{align*}
the inequality holds as $H$ is a histogram.
In the case {\it (ii)} we have
\begin{align*}
\sum \big[H-S\big](i,1\ldots j) \ = & \ 
\sum H(i,1\ldots j) - 1 \geq 0\\
= & \ \sum \big[H-S\big](i+1,1\ldots j+1);
\end{align*}
\noindent
the inequality holds due to definition of the function $x$.
In the case {\it (iii)} we have
\begin{align*}
\sum \big[H-S\big](i,1\ldots j)= 
0 \geq 0
= \sum  \big[H-S\big](i+1,1 \ldots j+1).
\end{align*}
\noindent
Thus $H-S$ is a histogram, which allows us to apply the induction assumption for $s-1.$
\end{proof}

\medskip

\begin{proof}[Proof of Lemma~\ref{lem:simplehistsum}]
Simple $\matr{r}{c}$-histograms are in one-to-one correspondence with monotonic functions $\set{1\ldots r} \to \set{1 \ldots c}$.
Indeed, such a function $f$ induces a simple histogram $S$ with $S(i, j) = 1$ if $f(i) = j$, and $S(i, j) = 0$ otherwise; on the other hand
a simple histogram $S$ defines a function that maps $i$ to the first (and the only) index $j$ with $S(i, j) = 1$, and this function is necessarily monotonic.

Fix a matrix $M$ of row dimension $r$ and of column dimensions $c$.
Consider a simple $\matr{c}{c'}$-histogram $S$ and the corresponding monotonic function $f : \set{1\ldots c} \to \set{1\ldots c'}$.
The multiplication $M \cdot S$ yields a $\matr{r}{c'}$-matrix whose $f(i)$-th column equals the $i$-th column of $M$, for $i = 1\ldots c$, and all other columns
are zero ones. Thus $M\cdot S \in \zeroext{M}$. Moreover, every $\matr{r}{c'}$-matrix $N \in\zeroext{M}$ is obtained in the same way via some 
monotonic function $f : \set{1\ldots c}\to\set{1\ldots c'}$, and thus $N = M \cdot S$ where $S$ is the corresponding simple $\matr{c}{c'}$-histogram.
\end{proof}

\medskip

\begin{proof}[Proof of \Cref{lem:conEqlin}]
The $\Qplus$-solution set of a homogenous system of equations is 
a finitely generated cone  (cf.~\cite{taming}, Prop.3).
Thus every transition in $\cal V$ can be expressed as a $\Qplus$-sum of a set of basic solutions,
and in consequence 
as a sequence of transitions 
of the form $c_i \vec v_i$, for $\vec v_i$ generators of the cone and $c_i \in \Qplus$, 
which are exactly transitions of a continuous PN. 

In the opposite direction, every transition of the continuous PN, as described  
above, will be also a transition of the homogeneous PN.
\end{proof}

\end{document}